%% file: main.tex
\documentclass[article]{llncs}
\input{preamble}

\title{Exact Algorithms for Edge Deletion to Cactus}

\author{Sheikh Shakil Akhtar \and
  Geevarghese Philip}

\institute{Chennai Mathematical Institute, Chennai, India}

\begin{document}

\maketitle

\begin{abstract}
  \input{abstract}
\end{abstract}

\section{Introduction}

\input{introduction}

\section{Preliminaries}

\input{preliminaries}

\section{Spanning tree to Cactus}
\input{tree-to-cactus}

\section{A Faster Exact Algorithm}
\input{new-exact-algo}

\section{Conclusion}
\input{conclusion}

\newpage

\bibliographystyle{splncs04}
\bibliography{references}

\newpage
\appendix
\section*{APPENDIX}

\input{appendix}

\end{document}

%% file: preamble.tex
\usepackage{xspace}
\usepackage{microtype}
\usepackage{framed}
\usepackage{xifthen}
\usepackage{tikz}
\usepackage{graphicx}
\usepackage{booktabs} 
\usepackage{hyperref}
\usepackage{mathtools}
\usetikzlibrary{calc}

\usepackage{algorithm}
\usepackage{algpseudocode}

\usepackage[capitalize,noabbrev]{cleveref}

\usepackage{amsmath}
\usepackage{amssymb}

\usepackage{enumitem}
\usepackage{thm-restate}

\newlength{\RoundedBoxWidth}
\newsavebox{\GrayRoundedBox}
\newenvironment{GrayBox}[1]%
   {\setlength{\RoundedBoxWidth}{.93\textwidth}
    \def\boxheading{#1}
    \begin{lrbox}{\GrayRoundedBox}
       \begin{minipage}{\RoundedBoxWidth}}%
   {   \end{minipage}
    \end{lrbox}
    \begin{center}
    \begin{tikzpicture}%
       \node(Text)[draw=black!20,fill=white,rounded corners,%
             inner sep=2ex,text width=\RoundedBoxWidth]%
             {\usebox{\GrayRoundedBox}};
        \coordinate(x) at (current bounding box.north west);
        \node [draw=white,rectangle,inner sep=3pt,anchor=north west,fill=white] 
        at ($(x)+(6pt,.75em)$) {\boxheading};
    \end{tikzpicture}
    \end{center}}     
\newenvironment{defproblemx}[2][]{\noindent\ignorespaces%
                                \FrameSep=6pt%
                                \parindent=0pt%
                \vspace*{-1.5em}
                \ifthenelse{\isempty{#1}}{%
                  \begin{GrayBox}{\textsc{#2}}%
                }{%
                  \begin{GrayBox}{\textsc{#2}  parameterized by~{#1}}%
                }
                \begin{tabular*}{\textwidth}{@{\hspace{.1em}} >{\itshape} p{1.8cm} p{0.8\textwidth} @{}}%
            }{
                \end{tabular*}%
                \end{GrayBox}%
                \ignorespacesafterend
            }  

\usepackage[textsize=large]{todonotes}


%% file: abstract.tex



We study two related problems on simple, undirected graphs: {\sc Edge Deletion
to Cactus} and {\sc Spanning Tree to Cactus}.

{\sc Edge Deletion to Cactus} has been known to be NP-hard on general graphs at
least since 1988. In their recent work where they showed---\emph{inter
alia}---that {\sc Edge Deletion to Cactus} is NP-hard even when restricted to
bipartite graphs, Koch et al. (Disc. Appl. Math., 2024) posed the complexity of
{\sc Spanning Tree to Cactus} as an open problem. Our first main result is that
{\sc Spanning Tree to Cactus} can be solved in polynomial time.


We use the polynomial-time algorithm for {\sc Spanning Tree to Cactus} as a
black box to obtain an exact algorithm that solves {\sc Edge Deletion to Cactus}
in \(\mathcal{O}^{\star}(n^{\,n-2})\) time, where \(n\) is the number of
vertices in the input graph and the \(\mathcal{O}^{\star}()\) notation hides
polynomial factors. This is a significant improvement over the brute-force
algorithm for {\sc Edge Deletion to Cactus} that enumerates all subsets of the
edge set.

We further improve the running time for {\sc Edge Deletion to Cactus} to
\(\mathcal{O}^{\star}(3^{n})\) using dynamic programming; this is our second
main result. To the best of our knowledge, this is the fastest known
exact-exponential algorithm for {\sc Edge Deletion to Cactus}.

%% file: introduction.tex
A \emph{cactus graph} is a connected graph in which every edge belongs to at most one cycle.
Equivalently, any two simple cycles in a cactus share at most one vertex.
Another equivalent characterization is that every block (biconnected component) of the graph is
either a single edge or a simple cycle.
This characterization implies that cactus graphs can be recognized in polynomial time using a block
decomposition algorithm, such as the classical linear-time algorithm for computing biconnected
components by Hopcroft and Tarjan \cite{10.1145/362248.362272}.
Cactus graphs form a natural generalisation of trees and arise in several areas of graph theory,
network design, and algorithmic graph modification problems \cite{Harary1953,BRIMKOV2017393}.

In this paper, we study two closely related edge modification problems whose objective is to
transform a given connected graph into a cactus by deleting or selecting edges.

\vspace{2.5mm}

\noindent\fbox{\begin{minipage}{\textwidth}
{\sc Spanning Tree To Cactus}\\
Input: A simple, undirected connected graph \(G = (V,E)\) and a spanning tree
\(T\) of \(G\).

Task: Find the maximum size of a set \(F \subseteq E \setminus E(T)\) such
that \(H = (V, E(T) \cup F)\) is a spanning cactus subgraph of \(G\).
\end{minipage}}

\vspace{2.5mm}

\noindent\fbox{\begin{minipage}{\textwidth}
{\sc Edge Deletion To Cactus}\\
Input: A simple, undirected connected graph \(G = (V,E)\).

Task: Find the size of a minimum sized set \(F \subseteq E\) such that
\(H = (V, E \setminus F)\) is a connected cactus subgraph of \(G\).
\end{minipage}}

\vspace{2.5mm}

The second problem, {\sc Edge Deletion to Cactus}, belongs to the broad class of graph modification
problems, in which the goal is to transform a given graph into a graph belonging to a specified
class using a minimum number of edits.
Such problems have been extensively studied for graph classes including forests, interval graphs,
chordal graphs, and planar graphs.
A closely related problem is {\sc Vertex Deletion to Cactus}, which asks for a minimum-sized set of
vertices whose removal transforms the input graph into a cactus.
Several results in this direction are
known \cite{TSUR2023106317,10.1007/978-3-662-53536-3_20,10.1007/978-3-642-34611-8_19,Aoike2022}.

The computational complexity of {\sc Edge Deletion to Cactus} was first investigated by
El-Mallah and Coulborn \cite{34972242c5ef4f499fd7580084406c9c}, who proved that the problem is
NP-complete for general graphs.
More recently, Koch et al. \cite{KOCH2024122} strengthened this result by showing that the problem
remains NP-complete even when restricted to bipartite graphs.
In the same work, the authors explicitly posed the complexity of the first problem,
{\sc Spanning Tree to Cactus}, as an open question.

To the best of our knowledge, no polynomial-time algorithm was previously known for
{\sc Spanning Tree to Cactus}, and no exact exponential-time algorithm was known for
{\sc Edge Deletion to Cactus} with a running time substantially better than brute-force enumeration
over all subsets of edges.

\medskip

\paragraph{Our contributions.}
We resolve both problems algorithmically.

First, we design a polynomial-time algorithm for {\sc Spanning Tree to Cactus}.

\begin{restatable}[]{theorem}{treetocactus}
\label{thm:treetocactus}
There exists a polynomial-time algorithm that solves {\sc Spanning Tree to Cactus}.
\end{restatable}

This algorithm is based on a reduction to the Maximum Independent Set problem on Edge Path Tree (EPT)
graphs \cite{GOLUMBIC19858,GOLUMBIC1985151,TARJAN1985221}, obtained by constructing
an auxiliary intersection graph of
the unique paths between the endpoints of the non-tree edges.
By combining this polynomial-time routine with enumeration over all spanning trees of \(G\), we
obtain our first exact algorithm for {\sc Edge Deletion to Cactus}.
(In the following theorem the notation \(\mathcal{O}^{\star}(.)\) hides
the polynomial factors on \(n\).)

\begin{restatable}[]{theorem}{ndeletiontocactus}
\label{thm:ndeletiontocactus}
There exists an algorithm that solves {\sc Edge Deletion to Cactus} in time
\(\mathcal{O}^{\star}(n^{\,n-2})\).
\end{restatable}

Although this algorithm is less efficient than the one described later, it still improves
substantially over the brute-force method of enumerating all subsets of the edge set, which would
require \(\mathcal{O}^{\star}(2^{n^2})\) time in the worst case.

Next, we solve the {\sc Edge Deletion to Cactus} problem independently in order to obtain a more
efficient exact algorithm.

\begin{restatable}[]{theorem}{deletiontocactus}
\label{thm:deletiontocactus}
There exists an algorithm that solves {\sc Edge Deletion to Cactus} in time
\(\mathcal{O}^{\star}(3^n)\).
\end{restatable}

\medskip

\paragraph{Organisation of the paper.}
We first establish structural properties of cactus graphs and edge-minimal non-cactus graphs.
Using these properties together with results on EPT graphs, we prove
\cref{thm:treetocactus}.
We then use this algorithm as a subroutine to obtain
\cref{thm:ndeletiontocactus}.

Finally, we develop additional structural observations on cactus graphs and use them to design a
dynamic programming algorithm that proves \cref{thm:deletiontocactus}.

%% file: preliminaries.tex
We use standard notations from graph theory.
All graphs considered in this paper are simple and undirected.

Let \(G = (V,E)\) be a graph.
Throughout the paper, we write \(n := |V|\) to denote the number of
vertices of the input graph. The \emph{size} of a graph is defined as \(|V| + |E|\).

For any graph \(H\), we denote its vertex set by \(V(H)\) and its edge
set by \(E(H)\).
For two graphs \(G_1\) and \(G_2\), their union, denoted by
\(G_1 \cup G_2\), is the graph with vertex set
\(V(G_1) \cup V(G_2)\) and edge set \(E(G_1) \cup E(G_2)\).

For a graph \(G\) and a vertex \(x \in V(G)\), we denote by
\(\deg_G(x)\) the degree of \(x\) in \(G\).
For any subset \(X \subseteq V(G)\), we denote by \(G[X]\) the subgraph
of \(G\) induced by \(X\).For any subset \(X \subseteq V(G)\), we denote by \(G - X\) the induced
subgraph \(G[V(G) \setminus X]\).
If \(X = \{v\}\) for some \(v \in V(G)\), we simply write \(G - v\).
For any subset \(F \subseteq E(G)\), we denote by \(G - F\) the graph
\((V(G), E(G) \setminus F)\).
If \(F = \{e\}\) for some \(e \in E(G)\), we simply write \(G - e\).

A \emph{block} of a graph is a maximal biconnected subgraph.

We now define Edge Path Tree (EPT) graphs following Golumbic and Jamison
\cite{GOLUMBIC19858,GOLUMBIC1985151}. Let \(T\) be a tree, and let \(\mathcal{P}\) be a collection of nontrivial
simple paths in \(T\); that is, each path in \(\mathcal{P}\) contains at
least one edge. We define the \emph{edge intersection graph} of \(\mathcal{P}\) in \(T\),
denoted by \(\Gamma(\mathcal{P}, T)\), as the graph whose vertex set
corresponds to the elements of \(\mathcal{P}\), with two vertices joined
by an edge if and only if their corresponding paths share at least one
edge in \(T\). A graph \(G\) is called an \emph{EPT graph} if there exist a tree \(T\)
and a collection of paths \(\mathcal{P}\) in \(T\) such that
\(G = \Gamma(\mathcal{P}, T)\).

%% file: tree-to-cactus.tex

\label{sec:treetocactus}

In this section, we first prove \cref{thm:treetocactus} by presenting a
polynomial-time algorithm for the problem {\sc Spanning Tree to Cactus}.
We then use this result to establish \cref{thm:ndeletiontocactus}, which
yields an algorithm for {\sc Edge Deletion to Cactus} with running time
\(\mathcal{O}^{\star}(n^{\,n-2})\).

We require the following theorem for the design of our algorithm.
Its proof is deferred to the appendix; see \cref{sec:edgeminimal}.
A pictorial illustration of the statement is given in
\cref{fig:edge-minimal-noncactus}.

\begin{restatable}[]{theorem}{cycleedges}
  \label{thm:cycleedges}
  Let \(G = (V,E)\) be an edge-minimal non-cactus connected graph; that is,
for every \(e \in E\), the graph \(G - e\) is a cactus.
Then, for any two edges \(e_1 = a_1 b_1\) and \(e_2 = a_2 b_2\), there
exists a cycle \(C\) in \(G\) that contains both \(e_1\) and \(e_2\).
\end{restatable}

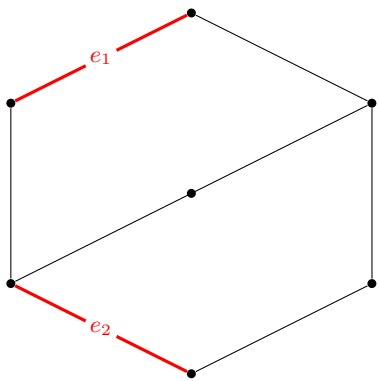
\begin{figure}[h]
\centering
\begin{tikzpicture}[
  scale=1.2,
  every node/.style={circle, fill=black, inner sep=1.2pt}
]

\node (a) at (0,2) {};
\node (b) at (2,3) {};
\node (c) at (4,2) {};
\node (d) at (4,0) {};
\node (e) at (2,-1) {};
\node (f) at (0,0) {};
\node (g) at (2,1) {};

\draw (a)--(b)--(c)--(d)--(e)--(f)--(a);

\draw (f)--(c);

\draw[very thick, red] (a)-- node[fill=white]{\(e_1\)}(b); 
\draw[very thick, red] (e)-- node[fill=white]{\(e_2\)}(f); 

\end{tikzpicture}
\caption{\small An edge-minimal non-cactus graph with two marked edges \(e_1\) and \(e_2\), such
that there exists a cycle which contains both \(e_1\) and \(e_2\), as claimed in \cref{thm:cycleedges}}
\label{fig:edge-minimal-noncactus}
\end{figure}

A lemma that is going to help us going forward is.

\begin{restatable}[]{lemma}{cyclicgraphtwo}
  \label{lm:cyclicgraph2}
  Let \(G\) be a connected graph.
  Suppose that there exist distinct vertices \(x, y \in V(G)\) and two
  distinct paths \(P_1\) and \(P_2\) with endpoints \(x\) and \(y\) such
  that
  \[
  E(G) = E(P_1) \cup E(P_2).
  \]
  Then every edge of \(G\) that belongs to \(E(P_1) \Delta E(P_2)\) lies
  in a cycle.
\end{restatable}

The following theorem establishes a key structural property that is
central to the design of our algorithm.

\begin{theorem}
  \label{th:treeplusedges}
  Let \(T = (V,E)\) be a tree, and let \(\hat{E}\) be a set of edges with
  endpoints in \(V\) such that \(\hat{E} \cap E = \emptyset\).
  Define a new graph \(H := (V, E \cup \hat{E})\).
  Then \(H\) is not a cactus if and only if there exist two edges in
  \(\hat{E}\) such that the unique paths between their respective
  endpoints in \(T\) share at least one edge.
\end{theorem}

\begin{proof}
Suppose first that there exist two edges in \(\hat{E}\), say \(e'_1\) and
\(e'_2\), such that the unique paths between their respective endpoints
in \(T\), denoted by \(P'_1\) and \(P'_2\), share at least one edge.
Then \(H\) contains two distinct cycles, namely \(P'_1 \cup e'_1\) and
\(P'_2 \cup e'_2\), that share an edge, and hence \(H\) is not a
cactus.

Conversely, suppose that \(H\) is not a cactus.
Let \(E^{\star} \subseteq \hat{E}\) be a minimal set such that the graph
\[
H^{\star} := (V, E \cup E^{\star})
\]
is not a cactus; that is, for every \(e \in E^{\star}\), the graph
\(H^{\star} - e\) is a cactus.
Clearly, \(H^{\star}\) is a subgraph of \(H\), and we must have
\(|E^{\star}| \geq 2\).

Assume, for a contradiction, that for every pair of edges in
\(\hat{E}\), and hence also in \(E^{\star}\), the unique paths between
their respective endpoints in \(T\) are edge-disjoint.
Let \(H'\) be an edge-minimal non-cactus subgraph of \(H^{\star}\); that is, for
every \(e \in E(H')\), the graph \(H' - e\) is a cactus.
Then \(H'\) must contain all edges of \(E^{\star}\), and therefore it
contains at least two edges from \(\hat{E}\); denote them by \(e_1\) and
\(e_2\).

Let \(P_1\) be the unique path in \(T\) between the endpoints of \(e_1\),
say \(a_1\) and \(b_1\), and let \(P_2\) be the unique path in \(T\)
between the endpoints of \(e_2\), say \(a_2\) and \(b_2\).
By \cref{thm:cycleedges}, there exists a cycle,
denoted by \(\tilde{C}\), in \(H'\) that contains both \(e_1\) and
\(e_2\).
Define
\[
\tilde{P} := \tilde{C} - e_1 .
\]
Then \(\tilde{P}\) is a path between \(a_1\) and \(b_1\) and contains the
edge \(e_2\).

Note that \(\tilde{P}\) is distinct from \(P_1\), since \(\tilde{P}\)
contains the edge \(e_2\), whereas \(P_1\) does not.
Consider the graph
\[
\hat{H} := \tilde{P} \cup P_1 .
\]
This graph contains the edge \(e_2\) but not \(e_1\).
Moreover, it is connected and every edge belongs to
\(\tilde{P}\) or to \(P_1\) , which are two distinct paths between
\(a_1\) and \(b_1\).
By \cref{lm:cyclicgraph2}, every edge in
\(E(\tilde{P}) \Delta E(P_1)\) lies in a cycle.

In particular, the edge \(e_2\) is contained in a cycle of
\(\hat{H}\); denote this cycle by \(\mathcal{C}^{\star}\).
This cycle is distinct from \(P_2 \cup \{e_2\}\), since it contains
edges from \(P_1\), whereas \(P_2 \cup \{e_2\}\) does not, as
\(E(P_1) \cap E(P_2) = \emptyset\) by assumption.

Observe that the cycle \(\mathcal{C}^{\star}\) is present in the graph
\[
H'' := (V, (E \cup E^{\star}) \setminus \{e_1\}).
\]
Thus, the graph \(H''\) contains two distinct cycles,
\(\mathcal{C}^{\star}\) and \(P_2 \cup \{e_2\}\), that share the edge
\(e_2\).
This contradicts the minimality of \(E^{\star}\), as \(H '' = H^{\star} - e_1\).
Therefore, there must exist two distinct edges in \(E^{\star}\) and hence in \(\hat{E}\), say
\(e_i\) and \(e_j\), such that the unique paths between their endpoints
share at least one edge.

\end{proof}

\subsection{The Algorithm}

\input{treetocactusalgo}

Observe that the above algorithm can be used to design an algorithm
for {\sc Edge Deletion To Cactus}, although not as efficient as the one
from \cref{thm:deletiontocactus}. Using any existing algorithm to enumerate
all possible spanning tree of the graph \(G\) (input to {\sc  Edge Deletion To Cactus}),
we call the above algorithm for each spanning tree of \(G\) and return the overall
maximum value reported. Winter \cite{Winter1986}, provides an
\(\mathcal{O}^{\star}(\tau(G))\) time algorithm for enumerating all
spanning trees of a graph \(G\), where \(\tau(G)\) is the number of spanning trees
of a graph \(G\). By the well-known Cayley's formula, \(\tau(G)\) can be
at most \(n^{n - 2}\). Thus, we obtain an \(\mathcal{O}^{\star}(n^{\,n-2})\)-time algorithm for
{\sc Edge Deletion to Cactus}, which improves upon the brute-force
approach of enumerating all subsets of the edge set of \(G\), an
approach that would require \(\mathcal{O}^{\star}(2^{n^{2}})\) time.

Thus, we have established the following theorem.

\ndeletiontocactus*

\begin{proof}
Let \(F \subseteq E\) be a minimum-size set of edges whose deletion transforms
\(G=(V,E)\) into a cactus, and let
\[
H := (V, E \setminus F).
\]
Then \(H\) is a connected cactus subgraph of \(G\).
In particular, \(H\) is a maximum-sized spanning cactus subgraph of \(G\).

Since \(H\) is connected, it contains a spanning tree \(T'\).
The tree \(T'\) is also a spanning tree of \(G\).

Starting from \(T'\), the graph \(H\) can be obtained by adding to \(T'\)
a set of non-tree edges such that the resulting graph remains a cactus.
By \cref{thm:treetocactus}, \cref{alg:treetocactus} computes,
for a given spanning tree \(T\), the maximum number of non-tree edges
that can be added to \(T\) while preserving the cactus property.
For the particular tree \(T'\), the algorithm therefore returns
\(|E(H)| - |E(T')|\).

Consequently, if we execute \cref{alg:treetocactus} for every
spanning tree of \(G\) and select the largest value obtained, we recover
the size of a maximum spanning cactus subgraph of \(G\), namely \(|E(H)|\).
This in turn, gives us the minimum size of the set of of edges to be deleted
from \(G\) such that it becomes a cactus, i.e., \(|E(G)| - |E(H)|\).

The number of spanning trees of an \(n\)-vertex graph is at most
\(n^{\,n-2}\).
Since \cref{alg:treetocactus} runs in polynomial time for each
spanning tree, the overall running time is
\(\mathcal{O}^{\star}(n^{\,n-2})\).

This yields an algorithm for {\sc Edge Deletion to Cactus} with running time
\(\mathcal{O}^{\star}(n^{\,n-2})\).
\end{proof}

%% file: treetocactusalgo.tex
Using \cref{th:treeplusedges}, we will design an algorithm for the
problem {\sc Spanning Tree To Cactus}.

Let \(G =(V, E)\) be a connected graph with \(T\) as its spanning tree,
which is given as input. The elements of \(E \setminus E(T)\) are precisely
the non-tree edges.

Construct an auxiliary graph \(\mathcal{A}_T\) as follows:

\begin{itemize}
\item The vertices of \(\mathcal{A}_T\) correspond to the unique paths in
  \(T\) associated with the non-tree edges. For example, for a non-tree
  edge \(e = ab\), there is a vertex in \(\mathcal{A}_T\) corresponding
  to the unique path in \(T\) between \(a\) and \(b\).

\item Two vertices in \(\mathcal{A}_T\) are adjacent if and only if their
  corresponding tree paths share at least one edge.
\end{itemize}

Observe that, for any set of non-tree edges \(F (\subseteq E \setminus E(T))\),
the graph \(H : = (V, E(T) \cup F)\) is a cactus if and only if the vertices of
\(\mathcal{A}_T\) corresponding to the set \(F\) forms an independent set. This is
due to \cref{th:treeplusedges}.

\begin{lemma}
  \label{lm:whencactus}
  Let \(G = (V,E)\) be a connected graph and let \(T\) be a spanning
  tree of \(G\). Let \(F \subseteq E \setminus E(T)\). We construct an auxiliary graph
  \(\mathcal{A}_T\) as shown above. Then the graph \(H : = (V, E(T) \cup F)\)
  is a cactus if and only if the set of vertices of \(\mathcal{A}_T\) corresponding to \(F\) forms
  an independent set in \(\mathcal{A}_T\).
\end{lemma}

\begin{proof}

  By \cref{th:treeplusedges}, \(H\) is a cactus if and only if, there exists
  no two edges from \(F\) such that the unique paths between their endpoints
  intersect. Thus, \(H\) is a cactus if and only if the  set of vertices of
  \(\mathcal{A}_T\) corresponding to \(F\) forms an independent set in \(\mathcal{A}_T\).
  
\end{proof}

By construction the graph \(\mathcal{A}_T\) is an EPT graph.
Moreover, Tarjan \cite{TARJAN1985221} showed that the size of a maximum
independent set in an EPT graph can be computed in time polynomial in the
size of the graph.

This yields a polynomial-time algorithm for our problem: the auxiliary
graph \(\mathcal{A}_T\) can be constructed in polynomial time, and the
maximum independent set of \(\mathcal{A}_T\) can then be computed in
polynomial time.

\begin{algorithm}
\caption{\textsc{SpanningTreeToCactus}\((G,T)\)}
\label{alg:treetocactus}
\begin{algorithmic}[1]
\Require A connected graph \(G=(V,E)\) and a spanning tree \(T\) of \(G\)
\Ensure The maximum number of non-tree edges that can be added to \(T\)
        such that the resulting graph remains a cactus

\State Construct the auxiliary graph \(\mathcal{A}_T\) as defined earlier
      (each vertex corresponds to a non-tree edge of \(G\), and two
      vertices are adjacent if and only if the corresponding tree paths in \(T\)
      share at least one edge)

\State Compute a maximum independent set of \(\mathcal{A}_T\) using the
algorithm of Tarjan \cite{TARJAN1985221}

\State \Return the size of this independent set
\end{algorithmic}
\end{algorithm}

\begin{lemma}
\label{lm:algo1}
\cref{alg:treetocactus} is correct and runs in polynomial time.
\end{lemma}

\begin{proof}
The correctness of \cref{alg:treetocactus} follows directly from
\cref{lm:whencactus}.

We now analyse its running time.
Let \(T\) be a spanning tree of the input graph \(G\).
The number of non-tree edges in \(G\) is at most \(\mathcal{O}(n^2)\).
For each such edge, the algorithm constructs a vertex in the auxiliary graph
\(\mathcal{A}_T\).
Hence, the size of \(\mathcal{A}_T\) is bounded by \(\mathcal{O}(n^2)\).

The edges of \(\mathcal{A}_T\) are determined by testing pairwise intersections
of the corresponding tree paths, which can be carried out in polynomial time.
Therefore, the construction of \(\mathcal{A}_T\) takes polynomial time in \(n\).

As \(\mathcal{A}_T\) is an EPT graph , the size of a maximum independent set in it can be computed
in time polynomial in the size of the graph \cite{TARJAN1985221}.
Since \(\mathcal{A}_T\) has size polynomial in \(n\), this step also runs in
polynomial time.

Consequently, \cref{alg:treetocactus} runs in time polynomial in \(n\).
\end{proof}

%% file: new-exact-algo.tex
In this section, we prove \cref{thm:deletiontocactus} and thereby obtain,
using dynamic programming, an improved algorithm for
{\sc Edge Deletion to Cactus} with running time
\(\mathcal{O}^{\star}(3^n)\).

We observe that solving the above problem is equivalent to finding the size of a
maximum-sized spanning cactus subgraph of \(G\).
Indeed, the complement of the edge set of such a spanning cactus is a
minimum-size set of edges whose removal transforms \(G\) into a cactus.
Thus, we will try to solve the problem of finding the size of the
largest spanning cactus of the input graph.

For \(X \subseteq V\), where \(G[X]\) is connected, we define \(I(X)\)
to be the number of edges in a maximum-sized spanning cactus subgraph
of \(G[X]\).

Let \(x \in X\) be such that \(\deg_{G[X]}(x) \geq 2\), and let
\(A \subseteq X \setminus \{x\}\) be a non-empty set.
Set
\(
B := X \setminus (A \cup \{x\}).
\)

Suppose that \(B\) is also nonempty and that both \(G[A \cup \{x\}]\) and
\(G[B \cup \{x\}]\) are connected.
We then define \(J_X[x,A]\) by
\[
\begin{aligned}
J_X[x,A] := \max \bigl\{\, |E(\mathcal{C})| \;\big|\;&
\mathcal{C} \text{ is a spanning connected subgraph of } G[X],\\
& x \text{ is a cut vertex of } \mathcal{C},\\
& \mathcal{C}[A \cup \{x\}] \text{ and }
\mathcal{C}[B \cup \{x\}] \text{ are cacti}
\,\bigr\}.
\end{aligned}
\]

Intuitively, \(J_X[x,A]\) is the maximum number of edges in a spanning
cactus subgraph of \(G[X]\) that is forced to decompose into two cactus
subgraphs on \(A \cup \{x\}\) and \(B \cup \{x\}\), joined together only
at the vertex \(x\).

Before proceeding further, we must ensure that \(J_X[x,A]\) is well
defined for all \(x \in X\) with \(\deg_{G[X]}(x) \geq 2\) and for all subsets
\(A \subseteq X \setminus \{x\}\) such that both \(A\) and \(B\) are nonempty,
where \(B := X \setminus (A \cup \{x\})\), and such that both
\(G[A \cup \{x\}]\) and \(G[B \cup \{x\}]\) are connected.
The idea is to construct separate spanning trees of the two connected
subgraphs \(G[A \cup \{x\}]\) and \(G[B \cup \{x\}]\), each rooted at
\(x\), and then to glue them together at \(x\).
Their union is a connected spanning subgraph of \(G[X]\) in which the
removal of \(x\) separates the graph into two forests, so that \(x\) is
a cut vertex. A more detailed proof is in the appendix (refer to
\cref{lm:cutvertex}).

We now proceed with the design of the algorithm to find the maximum
sized spanning cactus subgraph of a connected graph.

For every set \(X\), where \(G[X]\) is a cactus subgraph, we set
\[
I(X) := |E(G[X])|,
\] and for each \(x \in  X\) and \(A \subset X\), where \(J_X[x, A]\)
is defined, we set
\[
J_X[x, A] := |E(G[X])|.
\]

We now consider the remaining case in which \(G[X]\) is
connected but not a cactus. Note that, in such a case
\(| X | \geq 4\), as there are no non-cactus graphs with \(3\) or less vertices.
Before describing the algorithm for this case, we need to establish some
structural properties of cactus graphs that will be used in the
design of the algorithm.

Any connected non-cactus graph, must have a maximum sized
spanning cactus subgraph with a cut
vertex. If a maximum spanning cactus already has more than one block,
then it necessarily has a cut vertex.
Otherwise it is a simple cycle, and since the original graph is not a cactus,
we can swap in a non-cycle edge and remove a cycle edge to obtain another
spanning cactus of the same size that now has a cut vertex. 

\begin{lemma}
  \label{lm:lm1}
  Let \(H = (V,E)\) be a connected graph that is not a cactus.
  Then there exists a maximum-sized spanning cactus subgraph of \(H\)
  that has a cut vertex. (See \cref{fig:maxcactuscut} for an example.)
\end{lemma}

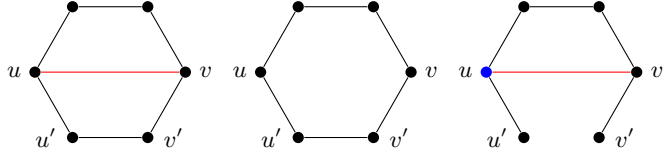
\begin{figure}[h]
\centering
\begin{tikzpicture}[scale=1, every node/.style={circle, fill=black, inner sep=1.5pt}]

\node[label=left:\(u'\)] (a1) at (0,0) {};
\node[label=right:\(v'\)] (a2) at (1,0) {};
\node[label=right:\(v\)] (a3) at (1.5,0.87) {};
\node (a4) at (1,1.74) {};
\node (a5) at (0,1.74) {};
\node[label=left:\(u\)] (a6) at (-0.5,0.87) {};

\draw (a1)--(a2)--(a3)--(a4)--(a5)--(a6)--(a1);
\draw[red] (a6)--(a3);

\node[label=left:\(u'\)] (b1) at (3,0) {};
\node[label=right:\(v'\)] (b2) at (4,0) {};
\node[label=right:\(v\)] (b3) at (4.5,0.87) {};
\node (b4) at (4,1.74) {};
\node (b5) at (3,1.74) {};
\node[label=left:\(u\)] (b6) at (2.5,0.87) {};

\draw (b1)--(b2)--(b3)--(b4)--(b5)--(b6)--(b1);

\node[label=left:\(u'\)] (c1) at (6,0) {};
\node[label=right:\(v'\)] (c2) at (7,0) {};
\node[label=right:\(v\)] (c3) at (7.5,0.87) {};
\node (c4) at (7,1.74) {};
\node (c5) at (6,1.74) {};
\node[fill=blue,label=left:\(u\)] (c6) at (5.5,0.87) {};

\draw (c2)--(c3)--(c4)--(c5)--(c6)--(c1);
\draw[red] (c6)--(c3);

\end{tikzpicture}
\caption{\small From left to right: the graph \(H\), a spanning cycle \(\mathcal{C}\) of
\(H\), and a spanning cactus subgraph \(\mathcal{C}'\) of \(H\) with cut
vertex \(u\), respectively.}
\label{fig:maxcactuscut}
\end{figure}

\begin{proof}
Let \(\mathcal{C}\) be a maximum-sized spanning cactus subgraph of \(H\).

If \(\mathcal{C}\) has more than one block, then it already has a cut
vertex and there is nothing to prove.
Otherwise, \(\mathcal{C}\) consists of a single block and is therefore a
simple cycle. Note that \(| V| \geq 4\).

Since \(H\) is not a cactus, there exists an edge \(e = uv \in E(H)\)
that is not contained in \(\mathcal{C}\).
Let \(e' = u'v' \in E(\mathcal{C})\).
Without loss of generality, we may assume that \(u \neq u'\).

Define
\[
\mathcal{C}' :=
\bigl(V,\; (E(\mathcal{C}) \cup \{e\}) \setminus \{e'\}\bigr).
\]
Then \(\mathcal{C}'\) is a spanning cactus subgraph of \(H\) with the
same number of edges as \(\mathcal{C}\).

We now show that \(u\) is a cut vertex of \(\mathcal{C}'\).
Removing the edge \(e'\) from the cycle \(\mathcal{C}\) turns it into a
path.
Adding the edge \(e = uv\) reconnects this path at the vertex \(u\) but
does not create a second internally disjoint path between the vertices
\(u\) and \(u'\).
Consequently, in the graph \(\mathcal{C}' - u\), the vertices that lie
on the former path between \(u\) and \(u'\) are separated from the rest
of the graph.

Thus, deleting \(u\) disconnects \(\mathcal{C}'\), and therefore \(u\)
is a cut vertex of \(\mathcal{C}'\).
\end{proof}

In the following lemma, we show that gluing two cactus graphs together at
a single common vertex preserves the cactus property.
Intuitively, every cycle in the resulting graph remains entirely within
one of the original cacti, and therefore no two cycles can share an
edge.

\begin{lemma}
	\label{lm:joincacti}
	Let \(\mathcal{C}\) and \(\mathcal{C}'\) be two edge-disjoint
	cacti, such that \(V(\mathcal{C}) \cap V(\mathcal{C}') = \{v\}\).
        Then \(\mathcal{C} \cup \mathcal{C}'\) is also a cactus.
\end{lemma}

\begin{proof}
        Any cycle in \(\mathcal{C} \cup \mathcal{C}'\), has all its edges
        either in \(\mathcal{C}\) or in \(\mathcal{C}'\). Thus, if two cycles
        in \(\mathcal{C} \cup \mathcal{C}'\) share edges, then it would mean
        either \(\mathcal{C}\) or \(\mathcal{C}'\) is not a cactus, which is
        a contradiction.
\end{proof}

The next two theorems capture the additive behaviour of optimal cactus subgraphs
when a connected graph is decomposed at a cut vertex. When such a
decomposition separates the graph into two connected parts, any maximal
cactus structure respecting this separation consists of optimal cactus
subgraphs on each part, and its size is the sum of their sizes.

\begin{theorem}
  \label{thm:thm1}
  Let \(G = (V,E)\) be a connected graph, and let \(X \subseteq V\) be a
  set such that \(G[X]\) is connected. Then
	\[
	I(X) = \underset{x \in X,\, A \subseteq X}{\emph{max}} J_X[x,A].
	\]
\end{theorem}

\begin{proof}
We first show that
\[
I(X) \leq \underset{x \in X,\, A \subseteq X}{\emph{max}} J_X[x,A].
\]

Let \(\mathcal{C}\) be a spanning cactus subgraph of \(G[X]\) with the
maximum possible number of edges. By definition,
\[
|E(\mathcal{C})| = I(X).
\]
By \cref{lm:lm1}, we can assume that the cactus \(\mathcal{C}\) has at least one cut vertex;
fix such a vertex and denote it by \(x\).
Let \(A\) be one of the connected components of \(\mathcal{C} - x\), and
set
\[
B := X \setminus (A \cup \{x\}).
\]

Definitely both \(A\) and \(B\) are non-empty and \(\deg_{G[X]}(x) \geq 2\).
Consider the induced subgraphs \(\mathcal{C}[A \cup \{x\}]\) and
\(\mathcal{C}[B \cup \{x\}]\).
These are cactus subgraphs of \(G[A \cup \{x\}]\) and
\(G[B \cup \{x\}]\), respectively, and hence both
\(G[A \cup \{x\}]\) and \(G[B \cup \{x\}]\) are connected.
Moreover, \(\mathcal{C}\) is exactly the union of
\(\mathcal{C}[A \cup \{x\}]\) and \(\mathcal{C}[B \cup \{x\}]\).
It follows that the number of edges of \(\mathcal{C}\) is bounded above
by \(J_X[x,A]\) for this choice of \(x\) and \(A\), and therefore
\[
I(X) \leq \underset{x \in X,\, A \subseteq X}{\emph{max}} J_X[x,A].
\]

We now prove the reverse inequality.
Let \(x' \in X\) and let \(A' \subseteq X \setminus \{x'\}\) be such that
\(A' \neq \emptyset\) and the induced subgraphs
\(G[A' \cup \{x'\}]\) and \(G[B' \cup \{x'\}]\) are connected, where
\[
B' := X \setminus (A' \cup \{x'\}).
\]

Let \(B' \neq \emptyset\).
Let \(\mathcal{C}'\) be a subgraph of \(G[X]\) with the maximum number of
edges among all spanning subgraphs of \(G[X]\) that have \(x'\) as a cut
vertex and for which \(\mathcal{C}'[A' \cup \{x'\}]\) and
\(\mathcal{C}'[B' \cup \{x'\}]\) are cacti.
By construction,
\[
|E(\mathcal{C}')|
=  J_X[x',A'].
\]
Since \(\mathcal{C}'\) is a cactus subgraph of \(G[X]\) by
\cref{lm:joincacti}, its number of edges is at most \(I(X)\).
And \(x'\) and \(A'\) were chosen arbitrarily, we have
\[
\underset{x \in X,\, A \subseteq X}{\emph{max}} J_X[x,A] \leq I(X).
\]

Combining the two inequalities completes the proof.
\end{proof}

\begin{theorem}
  \label{thm:thm2}
  Let \(G = (V,E)\) be a connected graph, and let \(X \subseteq V\) be
  such that \(G[X]\) is connected.
  Let \(x \in X\) satisfy \(\deg_{G[X]}(x) \geq 2\), and let
  \(A \subseteq X \setminus \{x\}\) be a nonempty set.
  Set
  \[
  B := X \setminus (A \cup \{x\}),
  \]
  and suppose that \(B\) is nonempty and that both
  \(G[A \cup \{x\}]\) and \(G[B \cup \{x\}]\) are connected.
  Then
  \[
  J_X[x,A] = I(A \cup \{x\}) + I(B \cup \{x\}).
  \]
\end{theorem}

\begin{proof}
Let \(\mathcal{C}\) be a maximum-sized connected subgraph of \(G[X]\)
such that \(x\) is a cut vertex of \(\mathcal{C}\), and such that
\(\mathcal{C}[A \cup \{x\}]\) and \(\mathcal{C}[B \cup \{x\}]\) are
cacti. By definition,
\[
|E(\mathcal{C})| = J_X[x,A].
\]
Since both \(\mathcal{C}[A \cup \{x\}]\) and
\(\mathcal{C}[B \cup \{x\}]\) are cacti, we have
\[
|E(\mathcal{C}[A \cup \{x\}])| \leq I(A \cup \{x\})
\quad \text{and} \quad
|E(\mathcal{C}[B \cup \{x\}])| \leq I(B \cup \{x\}).
\]
Therefore,
\begin{equation}
  \label{eq:ineq1}
J_X[x,A] \leq I(A \cup \{x\}) + I(B \cup \{x\}).
\end{equation}

Conversely, let \(\mathcal{C}'\) be a maximum-sized spanning cactus
subgraph of \(G[A \cup \{x\}]\), and let \(\mathcal{C}''\) be a
maximum-sized spanning cactus subgraph of \(G[B \cup \{x\}]\).
Define
\[
\mathcal{C}^{\star} := \mathcal{C}' \cup \mathcal{C}''.
\]
By \cref{lm:joincacti}, \(\mathcal{C}^{\star}\) is a cactus subgraph of
\(G[X]\) and spans \(G[X]\).
Moreover, deleting \(x\) from \(\mathcal{C}^{\star}\) disconnects the
graph, since there is no edge with one endpoint in
\(\mathcal{C}'\) and the other in \(\mathcal{C}''\).
Hence, \(x\) is a cut vertex of \(\mathcal{C}^{\star}\).

Thus \(|E(\mathcal{C}^{\star})| \leq J_X[x,A]\), and hence
\begin{equation}
  \label{eq:ineq2}
I(A \cup \{x\}) + I(B \cup \{x\}) \leq J_X[x,A].
\end{equation}

Combining the inequalities \ref{eq:ineq1} and \ref{eq:ineq2}, we
get the proof.
\end{proof}

Combining Theorems \ref{thm:thm1} and \ref{thm:thm2}, we obtain a
procedure for computing \(I(V)\), and hence the number of edges in a
maximum-sized spanning cactus subgraph of \(G\). Intuitively, the
algorithm performs dynamic programming over all vertex
subsets \(X \subseteq V\) that induce connected subgraphs, and for each
such set it considers all possible ways of choosing a cut vertex and
splitting the remaining vertices into two connected parts.
This leads to a branching behaviour in which each vertex may belong to
one of three roles (outside the current set, in the first part, or in
the second part), yielding a total running time of order \(3^n\).
A pseudocode description of the algorithm along with more
detailed proofs of the necessary theorems have been included in
the appendix(refer to \cref{sec:thedpalgo}).

Thus we have established the following theorem.

\deletiontocactus*

%% file: conclusion.tex
We studied two closely related graph modification problems whose objective is to transform a connected graph
into a cactus.

For the {\sc Edge Deletion to Cactus} problem, which was known to be NP-hard, we presented an exact algorithm
running in time
\(\mathcal{O}^{\star}(3^n)\).
We also resolved the previously open problem {\sc Spanning Tree to Cactus} by designing a polynomial-time
algorithm based on an auxiliary EPT graph construction and the computation of a maximum independent set.
As a consequence, this result yields an alternative exact algorithm for {\sc Edge Deletion to Cactus} with
running time
\(\mathcal{O}^{\star}(n^{\,n-2})\).
Although this algorithm is less efficient than our \(\mathcal{O}^{\star}(3^n)\) algorithm, it still improves
substantially over brute-force enumeration of all subsets of edges. An obvious
open question would be to obtain an algorithm for {\sc Edge Deletion to Cactus} which runs in time
\(\mathcal{O}^{\star}(c^n)\), for some \(c < 3\).

Another important direction for future research is the parameterised version of {\sc Edge Deletion to Cactus}, where
the parameter \(k\) is the number of deleted edges.
We observe that if an input graph \(G\) is a {\sc Yes}-instance of the parameterised version of
{\sc Edge Deletion to Cactus}, then the treewidth of \(G\) is at most \(2k+2\)~\cite{KOCH2024122}.
Since cactus graphs are exactly those graphs that exclude the diamond as a minor, and since minor-closed
graph classes are expressible in monadic second-order logic, the problem admits a fixed-parameter
tractable algorithm by Courcelle’s Theorem~\cite{zbMATH05852772}.
However, it remains an open challenge to design an explicit FPT algorithm with a practical and
small dependence on \(k\).

Further directions include weighted variants of the problem, where edges have costs and the goal
is to minimise the total deletion cost, as well as extensions to related graph modification problems
targeting cactus-like graph classes.

We hope that the structural insights and algorithmic techniques developed in this work will
stimulate further progress on cactus modification problems and related graph editing questions.

%% file: appendix.tex
\section{Properties of Edge Minimal Cacti}
 \label{sec:edgeminimal}
 \input{cycle-containing-edges}

 \section{Relationship between Trees and Cacti}
 \label{sec:relationship}
 \input{maintreetheorem}

\section{Some results needed for the Dynamic Programming Algorithm}

\begin{lemma}
  \label{lm:cutvertex}
Let \(H = (V,E)\) be a connected graph, and let \(x \in V\) be a vertex
such that \(\deg_H(x) \geq 2\). Let \(A \subseteq V \setminus \{x\}\) be nonempty, and set
\(B := V \setminus (A \cup \{x\})\), with \(B\) also nonempty.
Suppose that both \(H[A \cup \{x\}]\) and \(H[B \cup \{x\}]\) are
connected. Then there exists a spanning connected subgraph of \(H\),
say \(T\), such that \(x\) is one of the cut vertices of \(T\), and such
that both \(T[A \cup \{x\}]\) and \(T[B \cup \{x\}]\) are cacti.
\end{lemma}

\begin{proof}
Let \(x \in V\) be a vertex with \(\deg(x)_H \geq 2\), and let
\(A \subseteq V \setminus \{x\}\) be nonempty. Set
\[
B := V \setminus (A \cup \{x\}),
\]
and suppose that \(B\) is also nonempty. Assume further that both
\(H[A \cup \{x\}]\) and \(H[B \cup \{x\}]\) are connected.

Let \(T_1\) be a spanning tree of \(H[A \cup \{x\}]\) rooted
at \(x\), and let \(T_2\) be a spanning tree of
\(H[B \cup \{x\}]\) rooted at \(x\).
Define
\[
T := T_1 \cup T_2.
\]
By construction, \(T\) is a spanning subgraph of \(H\), and it is
connected. Moreover, deleting \(x\) from \(T\) disconnects the graph \(T\),
and hence \(x\) is a cut vertex of \(T\).
Finally, we observe that \(T[A \cup \{x\}] = T_1\) and
\(T[B \cup \{x\}] = T_2\), which are both trees and therefore cacti.
\end{proof}

\subsection{The Dynamic Programming Algorithm}
\label{sec:thedpalgo}

We describe a dynamic programming procedure for computing the value
\(I(X)\) for every vertex set \(X \subseteq V\) such that the induced
subgraph \(G[X]\) is connected.

For all sets \(X \subseteq V\) with \(|X| \leq 5\) and with \(G[X]\)
connected, the values \(I(X)\) and all defined values \(J_X[x,A]\) are
computed by brute force. For every vertex set \(X \subseteq V\) such that the induced subgraph
\(G[X]\) is a cactus (a property that can be tested in polynomial time),
we assign
\[
I(X) := |E(G[X])|.
\]
For all remaining sets \(X\), we initialise \(I(X) := 0\).

For each vertex set \(X \subseteq V\) such that \(G[X]\) is connected,
we invoke the procedure \textsc{FindMaxCactus}\((X)\) and assign its
return value to \(I(X)\).

\begin{algorithm}
\caption{\textsc{FindMaxCactus}\((X)\)}
\begin{algorithmic}[1]
\If{\(|X| \leq 5\) or \(G[X]\) is a cactus}
  \State \Return \(I(X)\)
  \EndIf

\ForAll{\(x \in X\) such that \(\deg_{G[X]} \geq 2\)}
  \ForAll{\(A \subseteq X \setminus \{x\}\)}
    \State \(B \gets X \setminus (A \cup \{x\})\)
    \If{\(A \neq \emptyset\) and \(B \neq \emptyset\)}
      \If{\(G[A \cup \{x\}]\) and \(G[B \cup \{x\}]\) are connected}
        \State \(s \gets \textsc{FindCutCactus}(x,A,X)\)
        \If{\(s > I(X)\)}
          \State \(I(X) \gets s\)
        \EndIf
      \EndIf
    \EndIf
  \EndFor
\EndFor

\State \Return \(I(X)\)
\end{algorithmic}
\end{algorithm}

The procedure \textsc{FindCutCactus}\((x,A,X)\) computes the contribution
of a decomposition of \(X\) into the two vertex sets
\(A \cup \{x\}\) and \(B \cup \{x\}\), where \(B := X \setminus (A \cup \{x\})\),
under the assumption that both induced subgraphs
\(G[A \cup \{x\}]\) and \(G[B \cup \{x\}]\) are connected. In other words, it calculates
the value \(J_X[x, A]\), for the choice of \(x\) and \(A\).

\begin{algorithm}
\caption{\textsc{FindCutCactus}\((x,A,X)\)}
\begin{algorithmic}[1]
\State \(B \gets X \setminus (A \cup \{x\})\)

\If{\(|A \cup \{x\}| \leq 5\)}
  \State \(s_1 \gets I(A \cup \{x\})\)
\Else
  \State \(s_1 \gets \textsc{FindMaxCactus}(A \cup \{x\})\)
\EndIf

\If{\(|B \cup \{x\}| \leq 5\)}
  \State \(s_2 \gets I(B \cup \{x\})\)
\Else
  \State \(s_2 \gets \textsc{FindMaxCactus}(B \cup \{x\})\)
\EndIf

\State \Return \(s_1 + s_2\)
\end{algorithmic}
\end{algorithm}


Intuitively, \textsc{FindMaxCactus}\((X)\) considers all choices of a cut
vertex \(x \in X\) of a cactus subgraph of \(G[X]\)
and all partitions of \(X \setminus \{x\}\) into two
nonempty sets \(A\) and \(B\) such that the induced subgraphs
\(G[A \cup \{x\}]\) and \(G[B \cup \{x\}]\) are connected.
For each such choice, it combines optimal cactus solutions on the two
vertex sets \(A \cup \{x\}\) and \(B \cup \{x\}\).
The procedure \textsc{FindCutCactus}\((x,A,X)\) computes the maximum size of a
spanning subgraph obtained by joining optimal cactus subgraphs of
\(G[A \cup \{x\}]\) and \(G[B \cup \{x\}]\) at the common vertex \(x\).
The correctness of the two procedures follows from \cref{thm:thm1} and
\cref{thm:thm2}.

We now analyse the running time of the procedures
\textsc{FindMaxCactus} and \textsc{FindCutCactus}.

\begin{theorem}
\label{thm:runtime}
Let \(G=(V,E)\) be a graph with \(n := |V|\).
The procedures \textsc{FindMaxCactus} and \textsc{FindCutCactus} compute
the values \(I(X)\) for all vertex sets \(X \subseteq V\) such that
\(G[X]\) is connected in time
\[
\mathcal{O}^{\star}(3^n).
\]
\end{theorem}

\begin{proof}
The dynamic programming table is indexed by vertex sets
\(X \subseteq V\) such that the induced subgraph \(G[X]\) is connected.
In the worst case, every subset of \(V\) induces a connected subgraph,
and hence the number of subproblems is at most \(2^n\).

Consider a fixed set \(X \subseteq V\) with \(|X| = k\).
If \(|X| \leq 5\) or \(G[X]\) is a cactus, the value \(I(X)\) is returned
in polynomial time.
Otherwise, the procedure \textsc{FindMaxCactus}\((X)\) iterates over all
vertices \(x \in X\) with \(\deg_{G[X]}(x) \geq 2\) and over all subsets
\(A \subseteq X \setminus \{x\}\).
For each such choice, it performs polynomial-time tests for
nonemptiness and connectivity of the induced subgraphs
\(G[A \cup \{x\}]\) and \(G[B \cup \{x\}]\), where
\(B := X \setminus (A \cup \{x\})\), and makes a call to
\textsc{FindCutCactus}\((x,A,X)\).

Thus, for a fixed set \(X\) of size \(k\), the total work performed by
\textsc{FindMaxCactus}\((X)\) is bounded by
\[
\mathcal{O}^{\star}(k \cdot 2^{k-1}) = \mathcal{O}^{\star}(2^k).
\]

The procedure \textsc{FindCutCactus}\((x,A,X)\) invokes the procedure
\textsc{FindMaxCactus} on the strictly smaller vertex sets
\(A \cup \{x\}\) and \(B \cup \{x\}\), and otherwise performs only
polynomial-time operations.
Hence, the running time of \textsc{FindCutCactus} is dominated by the
running time of these calls to \textsc{FindMaxCactus}.
In particular, \textsc{FindCutCactus} does not contribute any additional
exponential factor beyond that already accounted for in the analysis of
\textsc{FindMaxCactus}.

Summing over all vertex sets \(X \subseteq V\), the total running time is
bounded by
\[
\mathcal{O}^{\star}(\sum_{X \subseteq V} 2^{|X|})
= \mathcal{O}^{\star}\!\left( \sum_{k=0}^{n} \binom{n}{k} 2^{k} \right).
\]
Using the binomial identity
\(\sum_{k=0}^{n} \binom{n}{k} 2^{k} = (1+2)^{n} = 3^{n}\),
we obtain a total running time of \(\mathcal{O}^{\star}(3^{n})\).

Therefore, the combined execution of \textsc{FindMaxCactus} and
\textsc{FindCutCactus} runs in time \(\mathcal{O}^{\star}(3^{n})\).
\end{proof}

This proves the following theorem.

\deletiontocactus*

%% file: cycle-containing-edges.tex
\cycleedges*

Before proving the above theorem, we first establish a
few observations that will be used in the proof.

\begin{lemma}
  \label{lm:thirdcycle}
  Let \(G = (V,E)\) contain two distinct cycles \(C_1\) and \(C_2\) that
  share at least one edge. Then there exists a third cycle that is
  distinct from both \(C_1\) and \(C_2\). (As shown in \cref{fig:thirdcycle}.)
\end{lemma}

\begin{figure}[h]
\centering
\begin{tikzpicture}[
  scale=1.2,
  every node/.style={circle, fill=black, inner sep=1.2pt}
]

\node [label=left:\(u_1\)](u1) at (0,1) {};
\node [label=below:\(P\)](u2) at (1,2) {};
\node [label=left:\(u_3\)](u3) at (2,1) {};
\node [label=left:\(C_1\)](u4) at (1,3) {};
\node (u5) at (3,2) {};
\node [label=right:\(C_2\)](u6) at (3,0) {};

\draw[thick,red] (u1) -- (u2) -- (u3);

\draw (u3) -- (u4) -- (u1);

\draw (u3) -- (u5) -- (u6) -- (u1);

\end{tikzpicture}
\caption{\small Two distinct cycles \(C_1\) and \(C_2\) share the path \(P\).
Removing the edges of \(P\) leaves two distinct paths between \(u_1\) and \(u_3\),
which together form a third cycle \(C_3\).}
\label{fig:thirdcycle}
\end{figure}
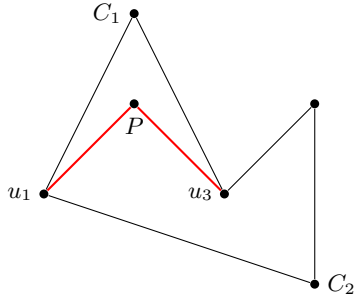

\begin{proof}
Consider a maximal connected subgraph of \(G\), say \(P\), whose edges
are all common to both \(C_1\) and \(C_2\).
Then \(P\) is a path in \(G\).
Write
\[
P = u_1 u_2 \ldots u_{\ell} .
\]
Note that both \(u_1\) and \(u_{\ell}\) belong to \(V(C_1) \cap V(C_2)\).

Now consider the graph
\[
G - \{u_1u_2, u_2u_3, \ldots, u_{\ell-1}u_{\ell}\},
\]
that is, the subgraph of \(G\) obtained by deleting the edges of \(P\).
This graph contains two paths between \(u_1\) and \(u_{\ell}\), namely
\[
P_1 := C_1 - E(P) \quad \text{and}
\quad P_2 := C_2 - E(P).
\]
These two paths must be distinct; otherwise, we would have
\(C_1 = C_2\) in \(G\), contrary to assumption.

Hence, the graph \(G - \{u_1u_2, u_2u_3,
\ldots, u_{\ell-1}u_{\ell}\}\) contains a
cycle.
Any such cycle is necessarily distinct from both \(C_1\) and \(C_2\).
\end{proof}

\begin{lemma}
  \label{lm:struct1}
  Let \(G = (V,E)\) be an edge-minimal non-cactus connected graph; that
  is, for every \(e \in E\), the graph \(G - e\) is a cactus.
  Let \(C_1\) and \(C_2\) be two cycles in \(G\) that share at least one
  edge. Then
  \[
  E \subseteq E(C_1) \cup E(C_2).
  \]
\end{lemma}

\begin{proof}
Suppose, for a contradiction, that there exists an edge
\(e' \in E\) such that \(e' \notin E(C_1) \cup E(C_2)\).
Consider the graph \(G - e'\).
In this graph, both cycles \(C_1\) and \(C_2\) are still present and
share at least one edge.
Hence, \(G - e'\) is not a cactus.
This contradicts the edge-minimality of \(G\).
\end{proof}

\begin{lemma}
  \label{lm:struct2}
Let \(G = (V,E)\) be an edge-minimal non-cactus connected graph.
Let \(C_1\) and \(C_2\) be two cycles in \(G\) that share at least one
edge.
Then, for every cycle \(C\) of \(G\), we have
\[
E(C) \cap E(C_1) \neq \emptyset.
\]
\end{lemma}

\begin{proof}
Suppose, for a contradiction, that there exists a cycle \(\hat{C}\) such
that
\[
E(\hat{C}) \cap E(C_1) = \emptyset.
\]
Then \(\hat{C}\) is distinct from both \(C_1\) and \(C_2\).
Consequently, there exists an edge \(\hat{e} \in E(\hat{C})\) such that
\(\hat{e} \notin E(C_2)\).
Moreover, \(\hat{e} \notin E(C_1)\), since
\(E(\hat{C}) \cap E(C_1) = \emptyset\).

Consider the graph \(G - \hat{e}\).
In this graph, the cycles \(C_1\) and \(C_2\) are still present and they
share at least one edge.
Hence, \(G - \hat{e}\) is not a cactus.
This contradicts the edge-minimality of \(G\).

Therefore, for every cycle \(C\) of \(G\), we must have
\(E(C) \cap E(C_1) \neq \emptyset\).
\end{proof}

\begin{lemma}
  \label{lm:struct3}
Let \(G = (V,E)\) be an edge-minimal non-cactus connected graph.
Let \(C_1\) and \(C_2\) be two cycles in \(G\) such that
\(E(C_1) \cap E(C_2) \neq \emptyset\), and let \(C_3\) be a distinct
cycle in \(G\).
Then
\[
E(C_1) \Delta E(C_2) \subseteq E(C_3).
\]
\end{lemma}

\begin{proof}
By \cref{lm:struct2}, we have
\[
E(C_1) \cap E(C_3) \neq \emptyset \text{ and } E(C_2) \cap E(C_3) \neq \emptyset
\]
Suppose, for a contradiction, that there exists an edge
\(e' \in E(C_1) \Delta E(C_2)\) such that \(e' \notin E(C_3)\).
There can be two cases.

\textbf{Case \(1\):} \(e' \in E(C_1)\).

Consider the graph \(G - e'\).
In this graph, the cycles \(C_2\) and \(C_3\) are both present and satisfy
\(E(C_2) \cap E(C_3) \neq \emptyset\).
Hence, \(G - e'\) is not a cactus.

\textbf{Case \(2\):} \(e' \in E(C_2)\).

Consider the graph \(G - e'\).
In this graph, the cycles \(C_1\) and \(C_3\) are both present and satisfy
\(E(C_1) \cap E(C_3) \neq \emptyset\).
Hence, \(G - e'\) is not a cactus.

Either case contradicts the edge-minimality of \(G\).

\end{proof}

\cycleedges*

\begin{proof}
Since \(G\) is a non-cactus, there exist two distinct cycles \(C_1\) and \(C_2\)
such that \(E(C_1) \cap E(C_2) \neq \emptyset\).
Consider two arbitrary edges \(e_1\) and \(e_2\).
By \cref{lm:struct1}, both \(e_1\) and \(e_2\) belong to
\(E(C_1) \cup E(C_2)\).
If both edges lie in \(E(C_1)\) or both lie in \(E(C_2)\), then there is
nothing to prove.

Now consider the case in which
\(e_1 \in E(C_1) \setminus E(C_2)\) and
\(e_2 \in E(C_2) \setminus E(C_1)\).
Since the cycles \(C_1\) and \(C_2\) share edges, \cref{lm:thirdcycle}
implies that there exists a third cycle \(C_3\) distinct from both
\(C_1\) and \(C_2\).
By \cref{lm:struct3}, we have
\[
E(C_1) \Delta E(C_2) \subseteq E(C_3).
\]
Therefore, there exists a cycle that contains both \(e_1\) and \(e_2\).
\end{proof}

%% file: maintreetheorem.tex
\begin{lemma}
  \label{lm:cyclicgraph1}
  Let \(G\) be a connected graph. Suppose that there exists distinct vertices
  \(x, y \in V(G)\) and two distinct edge-disjoint paths \(P_1\) and
  \(P_2\) with endpoints \(x\) and \(y\) such
  that
  \[
  E(G) = E(P_1) \cup E(P_2).
  \]
  Then every edge of \(G\) lies in a cycle. (See \cref{fig:onlycycles}.)
\end{lemma}

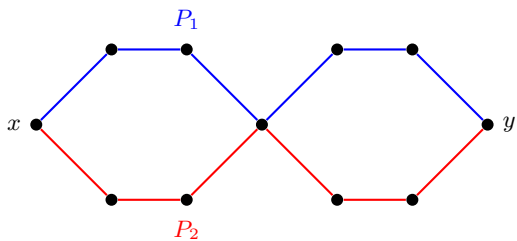
\begin{figure}[h]
    \centering
\begin{tikzpicture}[scale=1.0,
    vertex/.style={circle, fill=black, inner sep=1.6pt}]

\node[vertex,label=left:$x$] (x) at (0,0) {};
\node[vertex] (a1) at (1,1) {};
\node[vertex] (a2) at (2,1) {};
\node[vertex] (a3) at (3,0) {};

\node[vertex] (b1) at (1,-1) {};
\node[vertex] (b2) at (2,-1) {};

\node[vertex] (a4) at (4,1) {};
\node[vertex] (a5) at (5,1) {};
\node[vertex] (b4) at (4,-1) {};
\node[vertex] (b5) at (5,-1) {};

\node[vertex,label=right:$y$] (y) at (6,0) {};

\draw[thick, blue] (x) -- (a1) -- (a2) -- (a3) -- (a4) -- (a5) -- (y);
\node[blue] at (2,1.4) {$P_1$};

\draw[thick, red] (x) -- (b1) -- (b2) -- (a3) -- (b4) -- (b5) -- (y);
\node[red] at (2,-1.4) {$P_2$};

\end{tikzpicture}

\caption{\small A graph \(G\) with the properties as mentioned in the above lemma}
    \label{fig:onlycycles}
\end{figure}

\begin{proof}
Let \(e\) be an arbitrary edge of \(G\).
Since \(E(G) = E(P_1) \cup E(P_2)\), we have either \(e \in E(P_1)\) or
\(e \in E(P_2)\).
Without loss of generality, assume that \(e \in E(P_1)\), and let its
endpoints be \(u\) and \(v\).

Consider a maximal connected subgraph of \(P_1\) that contains the edge
\(e\) and whose endpoints, say \(w\) and \(z\) both belong to \(V(P_1) \cap V(P_2)\).
Without loss of generality, assume that \(w\) is connected to \(u\)
(or \(w = u\)) and that \(z\) is connected to \(v\) (or \(z = v\)) in
the graph \(P_1 - e\).
Moreover, \(w\) is connected to \(z\) by a path in \(P_2\).

We now construct a path \(P'\) between \(u\) and \(v\) in \(G - e\) as
follows:
\begin{itemize}
  \item a path from \(u\) to \(w\) in \(P_1 - e\), unless \(w = u\);
  \item a path from \(w\) to \(z\) in \(P_2\);
  \item a path from \(z\) to \(v\) in \(P_1 - e\), unless \(z = v\).
\end{itemize}

Thus, \(P' \cup \{e\}\) is a cycle that contains the edge \(e\).
Since \(e\) was chosen arbitrarily, every edge of \(G\) belongs to a
cycle.

\end{proof}


\cyclicgraphtwo*

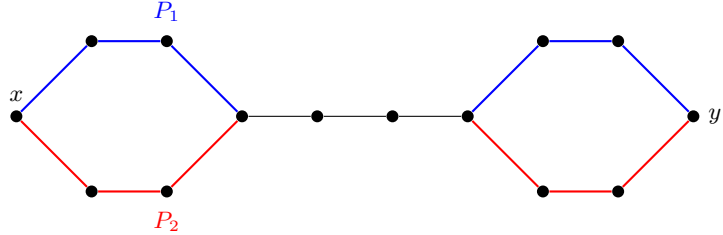
\begin{figure}[h]
    \centering
\begin{tikzpicture}[scale=1.0,
    vertex/.style={circle, fill=black, inner sep=1.6pt}]

\node[vertex,label=above:$x$] (x) at (0,0) {};
\node[vertex] (a1) at (1,1) {};
\node[vertex] (a2) at (2,1) {};
\node[vertex] (a3) at (3,0) {};

\node[vertex] (b1) at (1,-1) {};
\node[vertex] (b2) at (2,-1) {};

\node[vertex] (c1) at (4,0) {};
\node[vertex] (c2) at (5,0) {};
\node[vertex] (c3) at (6,0) {};
\node[vertex] (a4) at (7,1) {};
\node[vertex] (a5) at (8,1) {};
\node[vertex] (b4) at (7,-1) {};
\node[vertex] (b5) at (8,-1) {};

\node[vertex,label=right:$y$] (y) at (9,0) {};

\draw[thick, blue] (x) -- (a1) -- (a2) -- (a3);
\node[blue] at (2,1.4) {$P_1$};


\draw[thick, red] (x) -- (b1) -- (b2) -- (a3);
\node[red] at (2,-1.4) {$P_2$};

\draw[black] (a3) -- (c1) -- (c2) -- (c3);

\draw[thick, blue] (c3) -- (a4) -- (a5) -- (y);
\draw[thick, red] (c3) -- (b4) -- (b5) -- (y);

\end{tikzpicture}

\caption{\small A graph \(G\) with the properties as mentioned in the above lemma}
    \label{fig:nonlycycles}
\end{figure}

\begin{proof}
Let \(e\) be an arbitrary edge of \(G\) with
\(e \in E(P_1) \Delta E(P_2)\).
Without loss of generality, assume that \(e \in E(P_1)\), and let its
endpoints be \(u\) and \(v\).

Consider a maximal connected subgraph of \(P_1\) that contains the edge
\(e\) and contains no edges of \(P_2\); denote this subgraph by \(P'\).
Then \(P'\) is a subpath of \(P_1\).
Let its endpoints be \(w\) and \(z\).
By maximality, both \(w\) and \(z\) belong to \(V(P_1) \cap V(P_2)\).

Since \(w\) and \(z\) lie on both paths, there exists a path in \(P_2\)
between \(w\) and \(z\); denote this path by \(P''\).
By construction, we have
\[
E(P') \cap E(P'') = \emptyset,
\]
since \(P'\) contains no edges of \(P_2\).

Now consider the graph \(P' \cup P''\).
This graph is connected and has two distinct vertices \(w\) and \(z\)
joined by two edge-disjoint paths, namely \(P'\) and \(P''\).
Moreover, every edge of this graph belongs to \(E(P') \cup E(P'')\).
By \cref{lm:cyclicgraph1}, every edge of \(P' \cup P''\) lies in a cycle.

Since \(e \in E(P')\), it follows that \(e\) lies in a cycle.
As \(e\) was chosen arbitrarily, every edge in
\(E(P_1) \Delta E(P_2)\) lies in a cycle.
\end{proof}